\documentclass[12pt,a4paper]{article}
\usepackage[latin1]{inputenc}
\usepackage{amsmath}
\usepackage{amsthm}
\usepackage{amsfonts}
\usepackage{vmargin}
\usepackage{amssymb}
\usepackage{stmaryrd}
\author{Christopher Shirley}
\title{Decorrelation estimates for some continuous and discrete random schr\"{o}dinger operators in
dimension one, without covering condition}
\newtheorem{theo}{Theorem}[section]

\newtheorem{prop}[theo]{Proposition}
\newtheorem{lem}[theo]{Lemma}

\numberwithin{equation}{section}
 
\newcommand{\R}{\mathbb{R}}
\newcommand{\N}{\mathbb{N}}
\newcommand{\Z}{\mathbb{Z}}
\newcommand{\C}{\mathbb{C}}
\def\restriction#1#2{\mathchoice
              {\setbox1\hbox{${\displaystyle #1}_{\scriptstyle #2}$}
              \restrictionaux{#1}{#2}}
              {\setbox1\hbox{${\textstyle #1}_{\scriptstyle #2}$}
              \restrictionaux{#1}{#2}}
              {\setbox1\hbox{${\scriptstyle #1}_{\scriptscriptstyle #2}$}
              \restrictionaux{#1}{#2}}
              {\setbox1\hbox{${\scriptscriptstyle #1}_{\scriptscriptstyle #2}$}
              \restrictionaux{#1}{#2}}}
\def\restrictionaux#1#2{{#1\,\smash{\vrule height .8\ht1 depth .85\dp1}}_{\,#2}}

\begin{document}
\maketitle

\begin{abstract}
The purpose of the present work is to establish decorrelation estimates at distinct energies for some random Schrödinger operator in dimension one. In particular, we establish the result for some random operators on the continuum with alloy-type potential without covering condition assumption. These results are used to give a description of the spectral statistics.
\end{abstract}

\section{Introduction}
To  introduce our results, let us first consider one of the random operators that will be studied in the rest of this article. Let $(\omega_n)_{n\in\Z}$ be independent random variables, uniformly distributed on $[0,1]$ and define the random potential by $V_\omega(x)=\omega_n$ for $x\in(n-1/4,n+1/4)$ and $n\in\Z$ and zero elsewhere. The random potential is non-negative but not positive with probability one. Consider the operator $H_\omega: L^2(\R)\rightarrow  L^2(\R)$ defined by the following equation

\begin{equation}\label{simpleop}
\forall \phi\in \mathcal{H}^2(\R), H_\omega \phi=-\Delta \phi+V_\omega \phi.
\end{equation}

We know that, with probability one, $H_\omega$ is self-adjoint. As $H_\omega$ is $\Z$-ergodic, we know that there exists a set $\Sigma$ such that, with probability one, the spectrum of $H_\omega$ is equal to $\Sigma$ (see for instance \cite{CL90}). One of the purposes of this article is to give a description of the spectral statistics of $H_\omega$. In this context, we study the restriction of $H_\omega$ to a finite box and study the diverse statistics when the size of the box tends to infinity. For $L\in\N$, let $\Lambda_L=[-L,L]$ and $H_\omega(\Lambda_L)$ be the restriction of $H_\omega$ to $L^2(\Lambda_L)$ with Dirichlet boundary conditions. The spectrum of $H_\omega(\Lambda_L)$ is discrete and accumulate at $+\infty$. We denote $(E_j)_{j\in\N}$ the eigenvalues of $H_\omega(\Lambda)$, ordered increasingly and repeated according to multiplicity. 
We know from the $\Z$-ergodicity that there exists a deterministic, non-decreasing function $N$ such that, almost surely, we have
\begin{equation}
N(E)=\lim_{L\to\infty}\dfrac{\sharp\{j,E_j<E\}}{|\Lambda_L|}.
\end{equation}
The function $N$ is the integrated density of state (abbreviated IDS from now on), and it is the distribution function of a measure $dN$.

The main purpose of the present article is to prove the following theorem for a class of operators that contain the operator defined in \eqref{simpleop}.
\begin{theo}\label{joint}
There exists a discrete set $\mathcal{S}\subset (0,\infty)$ with no accumulation points such that for $(E_0,E_0')\in \R^2-\mathcal{S}^2$ such that $E_0\neq E_0'$ and such that $N(.)$ is differentiable at $E_0$ and $E_0'$ with $N'(E_0)>0$ and $N'(E_0')>0$, \\
when $|\Lambda|\rightarrow \infty$ the point processes $ \Xi(E_0,\omega,\Lambda)$ and $\Xi(E_0',\omega,\Lambda)$, converge weakly respectively to two independent Poisson processes on $\R$ with intensity the Lebesgue measure. That is, for any $(J_+,J_-)\in(\N^*)^2$, for any $(U_j^+)_{1\leq j\leq J_+}\subset\R^{J_+}$ and $(U_j^-)_{1\leq j\leq J_-}\subset\R^{J_-}$ collections of disjoint compact intervals, one has
\begin{displaymath}
\mathbb{P}\left(
\begin{aligned}
\sharp\{j;\xi_j(E_0,\omega,\Lambda)\in U_1^+\}&=k_1^+\\
\vdots\hspace*{8em} &\vdots\\
\sharp\{j;\xi_j(E_0,\omega,\Lambda)\in U_{J_+}^+\}&=k_{J_+}^+\\
\sharp\{j;\xi_j(E_0',\omega,\Lambda)\in U_1^-\}&=k_1^-\\
\vdots\hspace*{8em} &\vdots\\
\sharp\{j;\xi_j(E_0',\omega,\Lambda)\in U_{J_-}^-\}&=k_{J_-}^-
\end{aligned}
\right)\underset{|\Lambda|\to\infty}{\rightarrow} \prod_{j=1}^{J_+}\dfrac{|U_j^+|^{k_j^+}}{k_j^+!}e^{-|U_j^+|}\cdot\prod_{j=1}^{J_-}\dfrac{|U_j^-|^{k_j^-}}{k_j^-!}e^{-|U_j^-|}.
\end{displaymath}
\end{theo}

This theorem (with $\mathcal{S}=\emptyset$) was proved for the first time for a continuous model, the alloy-type model, in \cite{S14b} but only with a additional assumption on the random potential, the so-called covering condition, i.e when the bounded compactly supported, single site potential $q:\R\to\R$ generated by an atom at the origin satisfies the following inequality for some $\eta>0$
\begin{equation}\label{cov}
q\geq\dfrac{1}{\eta}\textbf{1}_{[-1/2,1/2]}.
\end{equation}
This assumption is not satisfied by the operator defined in \eqref{simpleop} and we will prove for the first time Theorem~\ref{joint} for a class of single site potentials that does not satisfy \eqref{cov}.
In order to study the spectral statistics of $H_\omega(\Lambda)$ and prove Theorem~\ref{joint} we use four results : the localization assumption, the Wegner estimates, the Minami estimates and the decorrelation estimates for distinct energies. As in \cite{S14b}, the three first assumptions are known to hold and we will prove the last one for the first time without the covering condition. We now introduce these assumptions.

Let $\mathcal{I}$ be an open relatively compact subset at the $\R$. We know from \cite{K14} that the operator satisfies the following localization assumption.

\textbf{(Loc): } for all $\xi\in(0,1)$, one has
\begin{equation}
\sup_{L>0} 
\underset{|f|\leq 1}{\underset{\text{supp }f\subset \mathcal{I}}\sup}
\mathbb{E}\left(\sum_{\gamma\in\Z^d}e^{|\gamma|^\xi}\|\textbf{1}_{[-1/2,1/2]}f(H_\omega(\Lambda_L))\textbf{1}_{[\gamma-1/2,\gamma+1/2]}\|_2\right)<\infty
\end{equation}
We know (see for instance \cite{CHK07}) that the following Wegner estimates hold on $\mathcal{I}$: 

\textbf{(W) :} There exists $C>0$,   such that for $J\subset \mathcal{I}$ and $L\in\N$ 
\begin{equation}
\mathbb{P}\Big[\text{tr} \left(\textbf{1}_J(H_\omega(\Lambda_L)) \right)\geq 1\Big ]\leq C |J||\Lambda_L|.
\end{equation}
This shows that the integrated density of state (abbreviated IDS from now on) $N(.)$ is Lipschitz continuous. As the IDS is a non-decreasing function, this implies that $N$ is almost everywhere differentiable and its derivative $\nu(.)$ is positive almost-everywhere on its essential support.

Let us now introduce the Minami estimates. We extract from \cite{K11} the
\begin{theo}[M]\label{mina-int}
Fix $J\subset\mathcal{I}$ a compact interval. For any $s'\in(0,1)$, $M>1$, $\eta>1$, $\rho\in(0,1)$, there exists $L_{s',M,\eta,\rho}>0$ and $C=C_{s',M,\eta,\rho}>0$ such that, for $E\in J$, $L\geq L_{s',M,\eta,\rho}$ and $\epsilon\in[L^{-1/s'}/M,ML^{-1/s'}]$ , one has
\begin{displaymath}
\sum_{k\geq2}\mathbb{P}\big(\textup{tr}[\textbf{1}_{[E-\epsilon,E+\epsilon]}(H_\omega(\Lambda_L))]\geq k\big)\leq C( \epsilon L )^{1+\rho}.
\end{displaymath}
\end{theo}

One purpose of this article is, as in \cite{GK10}, to give a description of spectral statistics. As (Loc), (W) and (M) hold, we know from \cite{GK10} that the following result hold. Define the \textit{unfolded local level statistics} near $E_0$ as the following point process :
\begin{equation}
\Xi(\xi;E_0,\omega,\Lambda)=\sum_{j\geq1} \delta_{\xi_j(E_0,\omega,\Lambda)}(\xi)
\end{equation}
 where
 \begin{equation}
 \xi_j(E_0,\omega,\Lambda)=|\Lambda|(N(E_j(\omega,\Lambda)-N(E_0)).
 \end{equation}
The unfolded local level statistics are described by the following theorem which is a weaker version of \cite[Theorem 1.9]{GK10}.

\begin{theo}\label{ULLS}
Pick $E_0\in \mathcal{I}$ such that $N(.)$ is differentiable at $E_0$ and $\nu(E_0)>0$.Then, when $|\Lambda|\to \infty$, the point process
$\Xi(\xi;E_0,\omega,\Lambda)$ converges weakly to a Poisson process with intensity the Lebesgue measure. That is, for any $p\in\N^*$, for any $(I_i)_{i\in\{1,\dots,p\}}$ collection of disjoint intervals
\begin{equation}
\lim_{|\Lambda|\to\infty}\mathbb{P}
\left(
\left\{\omega;
\begin{aligned}
\sharp\{j;\xi_j(\omega,\Lambda)\in I_1\}=k_1\\
\vdots\hspace*{8em} \vdots\hspace*{1em}\\
\sharp\{j;\xi_j(\omega,\Lambda)\in I_p\}=k_p
\end{aligned}
\right\}\right)=\dfrac{|I_1|^{k_1}}{k_1!}\dots\dfrac{|I_p|^{k_p}}{k_p!}
\end{equation}

\end{theo}

Therefore, Theorem~\ref{joint} answers the question about the joint behavior at large scale of the point processes $\Xi(\xi;E_0,\omega,\Lambda)$ and $\Xi(\xi;E_1,\omega,\Lambda)$ with $E_0\neq E_1$. Theorem~\ref{joint} is a weaker version of \cite[Theorem 1.10]{GK10} and it is proved in \cite{GK10} that it is a consequence of (Loc), (W), (M) and the decorrelation estimates at distinct energies, which are the following theorem.
\begin{theo}\label{decoInt}
There exists a discrete set $\mathcal{S}\subset\R$ and $\gamma>0$ such that for any $\alpha\in (0,1) $, $(E,E')\in(\R)^2-\mathcal{S}^2$ with $E\neq E'$ and $k>0$ there exists $C>0$ such that for $L$ sufficiently large and $kL^\alpha\leq l\leq L^\alpha/k$ we have 
\begin{displaymath}
\mathbb{P}\left(
\begin{aligned}
 \textup{tr}\, \textbf{1}_{[E-L^{-1},E+L^{-1}]}\left(H_\omega(\Lambda_l)\right)\neq 0,\\
 \textup{tr}\, \textbf{1}_{[E'-L^{-1},E'+L^{-1}]}\left(H_\omega(\Lambda_l)\right)\neq 0
\end{aligned}
\right)\leq C\dfrac{l^2}{L^{1+\gamma}}.
\end{displaymath}
\end{theo}
Theorem~\ref{decoInt} (with $\mathcal{S}=\emptyset$) was first proved in \cite{K11} for the discrete Anderson model in dimension one. Then, it was proved for other discrete models in dimension one in \cite{T14,S14a,S14b} and for the first time for a continuous model in \cite{S14b}, but with the covering condition. In the present article, we prove for the first time decorrelation estimates and therefore Theorom~\ref{joint}, without the covering condition.
\section{Models and Main result}
In this section, we introduce the models that will be studied,  the main result of this article and the known properties of the models used to prove this result. Let $(\omega_n)_{n\in\Z}$ be independent random variables with a common bounded, compactly supported density $\mu$.
\subsection{Models}
\subparagraph{Continuous models :}
Let $q:\R\to\R$ a continuous function such that there exist intervals $\mathcal{K}\subset\mathcal{J}$ of positive lengths and $C>0$ such that
\begin{equation}\label{stepfunction}
\frac{1}{C}\,\restriction{\textbf{1}}{\mathcal{K}}\leq q \leq C\, \restriction{\textbf{1}}{\mathcal{J}}\,.
\end{equation}
Therefore, $q$ is non-negative, bounded, compactly supported and positive on an interval of positive length. Contrary to the models studied in \cite{S14b}, the interval $\mathcal{K}$ is not supposed to be of length at least $1$, assumption that is often named \textit{covering condition}.

Let $H_\omega$ on $L^2(\R)$ defined by,
\begin{equation}
\forall\phi\in H^2(\R),\,H_\omega\phi=-\Delta\phi+q_{per}\phi+V_\omega\phi
\end{equation}
where $q_{per}$ is a bounded, one-periodic function and  
\begin{equation}
V_\omega(x)=\sum_{n\in\Z}\omega_n q(x-n).
\end{equation}
We suppose that the following hypothesis is true :
\medbreak
\textbf{(H)}: Either $q_{per}:=0$ or $q$ satisfies the covering condition, i.e $|\mathcal{K}|\geq 1$.
\medbreak
As $q$ is compactly supported and bounded, $V_\omega$ is uniformly bounded in $x$ and $\omega$. Therefore, we know that $H_\omega$ is self-adjoint with domain $H^2(\R)$ with probability one.
\subparagraph{Discrete models :}
Let $(a_n)_{n\in\Z}\in\left(\R_+\right)^\Z$ be a non-zero sequence of non-negative real numbers with finite support. Let $H_\omega$ on $\ell^2(Z)$ defined by,
\begin{equation}
\forall u\in \ell^2(\Z),\,H_\omega u=H_0u+V_\omega\phi
\end{equation}
where $H_0$ is a periodic bounded, Jacobi operator, and 
\begin{equation}
V_\omega(m)=\sum_{n\in\Z}\omega_n a_{m-n}.
\end{equation}
Note that $V_\omega$ and $H_\omega$ are uniformly bounded in $\omega$, hence $H_\omega$ is self-adjoint with probability one.

\subsection{Assumptions}

We know that the models defined above satisfy the three following assumptions for some relatively compact, open interval $\mathcal{I}\subset \R$.

\textbf{(W) :} There exists $C>0$ such that for $J\subset \mathcal{I}$ and $\Lambda$ an interval in $\R$, one has
\begin{equation}
\mathbb{P}\Big[tr \left(\textbf{1}_J(H_\omega(\Lambda)) \right)\geq 1\Big ]\leq C |J||\Lambda|.
\end{equation}
Wegner estimate has been proven for many different models, discrete or continuous (\cite{K95,CHK07,CGK09,V10}). Assumption (W) implies that the IDS is Lipschitz continuous.
\bigskip

\textbf{(Loc) : } for all $\xi\in(0,1)$, one has
\begin{equation}
\sup_{L>0} 
\underset{|f|\leq 1}{\underset{\text{supp }f\subset \mathcal{I}}\sup}
\mathbb{E}\left(\sum_{\gamma\in\Z^d}e^{|\gamma|^\xi}\|\textbf{1}_{\Lambda(0)}f(H_\omega(\Lambda_L))\textbf{1}_{\Lambda(\gamma)}\|_2\right)<\infty
\end{equation}
This property can be shown using either multiscale analysis or fractional moment method. In fact we suppose that $\mathcal{I}$ is a region where we can do the bootstrap multiscale analysis of \cite{GK01}. (Loc) is equivalent to the conclusion of the bootstrap MSA (see \cite[Appendix]{GK10} for details). We do not require estimates on the operator $H_\omega$ but only on $H_\omega(\Lambda_L)$. 
\bigskip

\textbf{(M) : }
Fix $J\subset \mathcal{I}$ a compact. For any $s'\in(0,1)$, $M>1$, $\eta>1$, $\rho\in(0,1)$, there exist $L_{s',M,\eta,\rho}>0$ and $C=C_{s',M,\eta,\rho}>0$ such that, for $E\in J,L\geq L_{s',M,\eta,\rho}$ and $\epsilon\in[L^{-1/s'}/M,ML^{-1/s'}]$ , one has
\begin{displaymath}
\sum_{k\geq2}\mathbb{P}\big(\textup{tr}[\textbf{1}_{[E-\epsilon,E+\epsilon]}(H_\omega(\Lambda_L))]\geq k\big)\leq C( \epsilon L )^{1+\rho}.
\end{displaymath}

The first two assumptions are known for a large class of operator, in any dimension. As for the last assumption, the Minami estimates, they are only proved in any dimension for Anderson type potential (\cite{Min96,GV07,BHS07,CGK09}), and for the discrete alloy-type model with single site potential whose Fourier transform does not vanish. To be more precise, these articles prove a stronger statement than the Minami estimates above, but this weaker version suffices in our case. 

For the models defined above, in dimension one, we know there exists a relatively compact, open interval $\mathcal{I}\subset \R$ such that (W), (Loc), (M) hold. It is proven in \cite{K14} that, in dimension one, for continuous models, if one has independence at a distance and localization, the Minami estimates (M) are an implication of the Wegner estimates. It is proven in \cite{S14a} that this statement holds also for discrete models, under the same assumptions. In both cases, the Minami estimates are not as strong as the Minami estimates proven in \cite{Min96,GV07,BHS07,CGK09}, but are sufficient for our purpose. For discrete alloy-type models, Minami estimates are also proven in \cite{TV14} but they only hold for single-site potentials whose Fourier transforms do not vanish. Therefore, we will use the Minami estimates proven in \cite{S14a} which hold under the assumptions of the present article.

\subsection{Main results}
The purpose of this article is to prove for the models defined above the following theorem
\begin{theo}\label{deco}
Suppose $H_\omega$ is one of the operators defined above. There exists $\gamma>0$ and a set $\mathcal{S}\subset\R$ with no accumulation point(only depending on $H_0$) such that, for any $\beta\in(1/2,1)$, $\alpha\in (0,1) $, $(F,G)\in\mathcal{I}^2-\mathcal{S}^2$ such that at $F\neq G$ and $k>0$, there exists $C>0$ such that for $L$ sufficiently large and $kL^\alpha\leq l\leq L^\alpha/k$ we have 
\begin{displaymath}
\mathbb{P}\left(
\begin{aligned}
 \textup{tr}\, \textbf{1}_{[F-L^{-1},F+L^{-1}]}\left(H_\omega(\Lambda_l)\right)\neq 0,\\
 \textup{tr}\, \textbf{1}_{[G-L^{-1},G+L^{-1}]}\left(H_\omega(\Lambda_l)\right)\neq 0
\end{aligned}
\right)\leq C\dfrac{l^2}{L^{1+\gamma}}.
\end{displaymath}
Furthermore, if $q$ satisfies the covering condition, $\mathcal{S}$ is the empty-set.
\end{theo}

Decorrelation estimates give more precise results about spectral statistics, such as Theorem~\ref{joint} (see \cite{GK10} for the proof and other results about spectral statistics). They are a consequence of Minami estimates and localization. In \cite{K11}, Klopp proves decorrelation estimates for eigenvalues of the discrete Anderson model in the localized regime. The result is proven at all energies only in dimension one. In \cite{T14}, decorrelation estimates are proven for the one-dimensional tight binding model, i.e when there are correlated diagonal and off-diagonal disorders. In \cite{S14a}, decorrelation estimates are also proven for other discrete models, such as Jacobi operators with positive alloy-type potential or the random hopping model, i.e when there is only off-diagonal disorder. Decorrelation estimates were also proved for continuous models in \cite{S14b} but only under the covering condition, and for the free Hamiltonian equal to the Laplace operator. In the present article, we improve this result by allowing a 1-periodic background potential. We also allow non-negative single-site potential without covering condition, but we then prove decorrelation estimates at all energies except for the ones in a fixed discrete set. The proof also apply to discrete operators but as the proof is the same, it will not be given. 

The proof of Theorem~\ref{deco} rely on the study of the gradients of two different eigenvalues. In particular, we show that the probability that they are co-linear is zero. In \cite{K11}, \cite{T14} and \cite{S14a}, this condition could easily be rewritten as a property of eigenvectors. For instance, for the discrete Anderson model, this condition is the system of equations
\begin{equation}
\forall n\in\llbracket -L,L \rrbracket, u^2(n)=v^2(n).
\end{equation}
 where $u$ and $v$ are normalized eigenvector associated to the eigenvalues. These equations can be rewritten easily as $u(n)=\pm v(n)$.
 
Now, consider the continuous alloy-type model where the single site potential $q$ has support included in $(0,1)$. Then, the condition of co-linearity is the system of equations
\begin{equation}\label{gradcolisimp}
\forall n\in\llbracket -L,L-1, \rrbracket, \int_n^{n+1}q(x)u^2(x)=\int_n^{n+1}q(x)v^2(x).
\end{equation}
The strategy developped in \cite{S14b} was to rewrite this system as a system of $2L$ quadratic equations, using basis of solutions on each interval $(n,n+1)$. This system and the fact that the eigenvectors have continuous derivatives will imposed conditions on the eigenvectors that are easier to handle. The difficulty was to choose a basis of solutions in which the problem could be rewritten in a simpler manner. The choice made in \cite{S14b} was to take orthonormal (with respect to $q$) basis of solution. We then had to compare the $L^2$ norms of these solutions when the $(\omega_n)_n$ were moving and the covering condition was simplifying this comparison. In the present article, we will make another choice of basis, the basis given by the Floquet theory. This will allow us remove the covering condition. The trade-off is that we need to exclude certain energies.
\section{Proof of Theorem~\ref{deco}}
We follow the proof of \cite[Section 2]{K11}. The only difference is in the proof of Lemma~\ref{probcoli} below which corresponds to \cite[Lemma 2.4]{K11}. The proof of the other intermediate results are the same as in \cite{K11}. Thus, the results will be given without proofs. The proof of Theorem~\ref{deco} is the same for discrete and continuous models except from the obvious modifications due to the discrete structure. Therefore, we will only prove the results for continuous models.

Using (M), Theorem~\ref{deco} is a consequence of the following theorem : 

\begin{theo}\label{thdec}
Let $\beta\in(1/2,1)$. For $\alpha\in (0,1) $ and $(F,G)\in\mathcal{I}^2$ with $F\neq G$, for any $k>1$ there exists $C>0$, such that for $L$ large enough and $kL^\alpha\leq l\leq L^\alpha/k$ we have
\begin{displaymath}
\mathbb{P}_0:=\mathbb{P}\left(
\begin{aligned}
  tr\textbf{1}_{[F-2L^{-1},F+2L^{-1}]}(H_\omega(\Lambda_l ))= 1,\\
 tr\textbf{1}_{[G-2L^{-1},G+2L^{-1}]}(H_\omega(\Lambda_l ))= 1
\end{aligned}
\right)\leq C\left(\dfrac{l^2}{L^{4/3}}\right)e^{(\log L)^\beta}.
\end{displaymath}
\end{theo}

We now restrict ourself to the study of the restriction of $H_\omega$ to cubes of size $(\log L)^{1/\xi'}$ instead of length $L^\alpha$. In this context, we extract from \cite[Proposition 2.1]{K11} the 
\begin{prop}\label{(Loc)(I)} : For all $p>0$ and $\xi\in(0,1)$, for L sufficiently large, there exists a set of configuration $\mathcal{U}_{\Lambda_l}$ of probability larger than $1-L^{-p}$ such that if $\phi_{n,\omega}$ is a normalized eigenvector associated to the eigenvalue $E_{n,\omega}\in\mathcal{I}$ and $x_0(\omega)\in \{1,\dots,L\}$ maximize $|\phi_{n,\omega}|$ then 
\begin{equation}\label{expdec}
|\phi_{n,\omega}(x)|\leq L^{p+d} e^{-|x-x_0|^{\xi}}.
\end{equation}
\end{prop}

Now, Theorem~\ref{thdec} is a consequence of the following lemma and Proposition~\ref{(Loc)(I)}. 
\begin{lem}\label{thdec2}
Let $\beta'\in(1/2,1)$. For $\alpha\in (0,1) $ and $(F,G)\in\mathcal{I}^2$ with $F\neq G$, there exists $C>0$ such that for any $\xi'\in(0,\xi)$, $L$ large enough and $\tilde{l}=(\log L)^{1/\xi'}$ we have 
\begin{displaymath}
\mathbb{P}_1:=\mathbb{P} \left( 
\begin{aligned}
 tr\textbf{1}_{[F-2L^{-1},F+2L^{-1}]}(H_\omega(\Lambda_{\tilde{l}} ))= 1,\\
 tr\textbf{1}_{[G-2L^{-1},G+2L^{-1}]}(H_\omega(\Lambda_{\tilde{l}} ))= 1
\end{aligned}
\right)\leq C\left (\dfrac{\tilde{l}^2}{L^{4/3}}\right)e^{\tilde{l}^{\beta'}}.
\end{displaymath}
\end{lem}
The rest of the section is dedicated to the proof of Lemma~\ref{thdec2}. Define $J_L=\left[E-L^{-1},E+L^{-1}\right]$ and $J_L'=\left[E'-L^{-1},E'+L^{-1}\right]$. For $\epsilon\in(2L^{-1},1)$, for some $\kappa>2$, using (M) when the operator $H_\omega(\Lambda_l)$ has two eigenvalues in $[-\epsilon,+\epsilon]$, one has 
\begin{equation}
\mathbb{P}_1\leq C\epsilon^2 l^{\kappa}+\mathbb{P_\epsilon}\leq C\epsilon^2l^2e^{l^\beta}+\mathbb{P}_\epsilon
\end{equation}
where 
\begin{displaymath}
\mathbb{P_\epsilon}=\mathbb{P}(\Omega_0(\epsilon))
\end{displaymath}
and
\[ \Omega_0(\epsilon)= \left\{ \omega;
\begin{aligned}
\sigma(H_\omega(\Lambda_l))\cap J_L&= \{E(\omega)\} \\
\sigma(H_\omega(\Lambda_l))\cap (E-\epsilon,E&+\epsilon)= \{E(\omega)\} \\
\sigma(H_\omega(\Lambda_l))\cap J_L'&= \{E'(\omega)\} \\
\sigma(H_\omega(\Lambda_l))\cap(E'-\epsilon,&E'+\epsilon)= \{E'(\omega)\}
\end{aligned} 
\right \}.
\]
In order to estimate $\mathbb{P}_\epsilon$ we make the following definition. For $(\gamma,\gamma')\in\Lambda_L^2$ let $J_{\gamma,\gamma'}(E(\omega),E'(\omega))$ be the Jacobian of the mapping $(\omega_\gamma,\omega_{\gamma'})\rightarrow (E(\omega),E'(\omega))$ : 
\begin{equation}\label{defjac}
J_{\gamma,\gamma'}(E(\omega),E'(\omega))=\left \vert \begin{pmatrix} \partial_{\omega_\gamma}E(\omega) & \partial_{\omega_{\gamma'}}E(\omega)\\ \partial_{\omega_\gamma}E'(\omega) &\partial_{\omega_{\gamma'}}E'(\omega)
\end{pmatrix} \right \vert
\end{equation} 
and define 
\begin{displaymath}
\Omega^{\gamma,\gamma'}_{0,\beta}(\epsilon)= \Omega_0(\epsilon)\cap \left \{ \omega ;|J_{\gamma,\gamma'}(E(\omega),E'(\omega))|\geq \lambda \right\}.
\end{displaymath} 

When one of the Jacobians is sufficiently large, the eigenvalues depends on  two independent random variables. Thus the probability to stay in a small interval is small. So we divide the proof in two parts, depending on whether all the Jacobians are small. The next lemma shows that if all the Jacobians are small then the gradients of the eigenvalues, which have positive components for the models considered in the present article, must be almost co-linear.
\begin{lem}\label{grad->jac}
Let $(u,v)\in(\R^+)^{2n}$ such that $\|u\|_1=\|v\|_1=1$. Then 
\begin{displaymath}
\max_{j\neq k} \left | \begin{pmatrix} u_j & u_k\\ v_j & v_k \end{pmatrix} \right |^2\geq \dfrac{1}{4n^5}\Vert u-v \Vert_1^2.
\end{displaymath}
\end{lem}
Thus, either one of the Jacobian determinants is not small or the gradient of $E$ and $E'$ are almost co-linear. We now show that the second case happens with a small probability. 
\begin{lem}\label{probcoli}
Let $(F,G)\in\mathcal{I}^2$ with $F\neq G$ and $\beta>1/2$. Furthermore, if $d>1$, we suppose that $|F-G|\geq \text{diam} sp(H_0)$. Let $\mathbb{P}$ denotes the probability that there exist $E_j(\omega)$ and $E_k(\omega)$, simple eigenvalues of $H_\omega(\Lambda_l)$ such that $|F-E_j(\omega)|+|G-E_k(\omega)|\leq e^{-l^\beta}$ and such that
\begin{equation}\label{gradcoli}
\left\|\dfrac{\nabla_\omega\big(E_j(\omega))}{\|\nabla_\omega\big(E_j(\omega))\|}-\dfrac{\nabla_\omega\big(E_k(\omega))}{\|\nabla_\omega\big(E_k(\omega))\|}\right\|\leq e^{-l^\beta}
\end{equation}
then there exists $c>0$ such that
\begin{equation}
\mathbb{P}\leq e^{-c l^{2\beta}}
\end{equation}
\end{lem}
 The proof of this result depends on the model and will be given below in the paper. First, we finish the proof of Lemma~\ref{thdec2}.
 
 Pick $\lambda=e^{-l^\beta}\|\nabla_\omega\big(E_j(\omega))\|\|\nabla_\omega\big(E_k(\omega))\|$. For the models considered in the present article, there exists $C>1$ such that for all $L$, $\|\nabla_\omega\big(E_j(\omega))\|\in[1/C,C]$. This will be proven in the following subsections. Therefore $\lambda\asymp e^{-l^\beta}$. Then, either one of the Jacobian determinant is larger than $\lambda$ or the gradients are almost co-linear. Lemma~\ref{probcoli} shows that the second case happens with a probability at most $e^{-cL^{2\beta}}$. It remains to evaluate $\mathbb{P}(\Omega^{\gamma,\gamma'}_{0,\beta}(\epsilon))$. We recall the following results from \cite{K11}. They were proved for the discrete Anderson model, they extend readily to our case. First, we study the variations of the Jacobian. 
\begin{lem}\label{Hessien}
There exists $C>0$ such that
\begin{displaymath}
\Vert Hess_\omega(E(\omega))\Vert_{l^\infty\rightarrow l^1}\leq \dfrac{C}{dist\big[E(\omega),\sigma(H_\omega(\Lambda_l))-\{E(\omega)\}\big ]}.
\end{displaymath}
\end{lem}
Fix $\alpha\in(1/2,1)$. Using Lemma~\ref{Hessien} and (M) when $H_\omega(\Lambda_l)$ has two eigenvalue in $[E-L^{-\alpha},E+L^{-\alpha}]$, for L large enough, with probability at least $1-L^{-2\alpha}\lambda $,
\begin{equation}
\Vert Hess_\omega(E(\omega))\Vert_{l^\infty\rightarrow l^1}+\Vert Hess_\omega(E'(\omega))\Vert_{l^\infty\rightarrow l^1}\leq CL^\alpha.
\end{equation}
 
In the following lemma we write $\omega=(\omega_\gamma,\omega_{\gamma'},\omega_{\gamma,\gamma'})$.
\begin{lem}\label{square}
Pick $\epsilon= L^{-\alpha}$. For any $\omega_{\gamma,\gamma'}$, if there exists $(\omega_\gamma^0,\omega_{\gamma'}^0)\in \R^2$ such that $(\omega_\gamma^0,\omega_{\gamma'}^0,\omega_{\gamma,\gamma'})\in\Omega^{\gamma,\gamma'}_{0,\beta}(\epsilon)$, then for $(\omega_\gamma,\omega_{\gamma'}) \in \R^2$ such that $|(\omega_\gamma,\omega_{\gamma'})-(\omega_\gamma^0,\omega_{\gamma'}^0)|_\infty\leq \epsilon$ one has 
\begin{displaymath}
(E_j(\omega),E_k(\omega))\in J_L\times J_L'\Longrightarrow |(\omega_\gamma,\omega_{\gamma'})-(\omega_\gamma^0,\omega_{\gamma'}^0)|_\infty\leq L^{-1}\lambda^{-2}.
\end{displaymath} 
\end{lem}

As in Lemma~\ref{square}, fix $(\omega_\gamma^0,\omega_{\gamma'}^0)$ such that $(\omega_\gamma^0,\omega_{\gamma'}^0,\omega_{\gamma,\gamma'})\in\Omega^{\gamma,\gamma'}_{0,\beta}(\epsilon)$ and define $\mathcal{A}:=(\omega_\gamma^0,\omega_{\gamma'}^0)+\{(\omega_\gamma,\omega_{\gamma'}) \in \R_+^2\cup \R_-^2 ,\left|\omega_\gamma\right|\geq \epsilon \text{ or } \left|\omega_{\gamma'}\right|\geq \epsilon \}$. We know that for any $i\in\Z$, $\omega_i\rightarrow E_j(\omega)$ and  $\omega_i\rightarrow E_k(\omega)$ are non increasing functions. Thus, if $ 
(\omega_\gamma,\omega_{\gamma'})\in\mathcal{A}$ then $(E_j(\omega),E_k(\omega))\notin J_L\times J_L'$. Thus, all the squares of side $\epsilon$ in which there is a point in $\Omega^{\gamma,\gamma'}_{0,\beta}(\epsilon)$ are placed along a non-increasing broken line that goes from the upper left corner to the bottom right corner. As the random variables are bounded by $C>0$, there are at most $C L^\alpha$ cubes of this type.

As the $(\omega_n)_n$ are i.i.d, using Lemma~\ref{square} in all these cubes, we obtain  : 
\begin{equation}\label{proba2}
\mathbb{P}(\Omega^{\gamma,\gamma'}_{0,\beta}(\epsilon))\leq CL^{\alpha-2}\lambda^{-4}
\end{equation}
and therefore
\begin{equation}
\mathbb{P}_\epsilon\leq CL^{\alpha-2}\lambda^{-3}.
\end{equation}
Optimization yields $\alpha=2/3$. This completes the proof of Theorem~\ref{thdec2}.

\section{Proof of Lemma~\ref{probcoli}}

In this section, we follow the strategy developed in \cite{S14b} but we first introduce some definitions. Recall that $q$ is the simple-site potential and that it satisfies \eqref{stepfunction}. On $L^2(-N,N)$ we define the non-negative symmetric bi-linear form : 
\begin{equation}\label{sem-inn}
\langle f,g\rangle_q = \int_{-N}^N f(t)g(t) q(t) dt.
\end{equation}
We denote $\|.\|_q$ the corresponding semi-norm. We say that the functions $f$ and $g$ are $q$-orthogonal if $\langle f,g\rangle_q=0$. The notion of 1-orthogonality is the usual orthogonality in $L^2(-l,l)$. Fix $(F,G)\in\R$ and let $u$ and $v$ be 1-normalized eigenfunctions of $H_\omega(\Lambda_l)$ associated to the eigenvalues $E_j(\omega)\in[F-e^{-l^\beta},F+e^{-l^\beta}]$ and $E_k(\omega)\in[G-e^{-l^\beta},G+e{-l^\beta}]$. These eigenvalues are almost surely simple and we compute  
\begin{equation}\label{dercont}
\partial_{\omega_n} E_j(\omega) =\left\langle \left(\partial_{\omega_n} H_\omega\right) u,u \right\rangle_1= \|u_{|_{(n-N,n+N)}}\|^2_q>0.
\end{equation}
First we show that the gradient of $E_j$ cannot be to small. To prove this, we restrict to our one dimensional setting the \cite[Theorem 2.1]{NTTV14}, which is a scale-free unique continuous principle, but we first introduce some notations.
For $\delta>0$ and $z:=(z_j)_{j\in\Z}$ a collection of point in $\Z$ such that $|z_j-j|\leq1$ define $\mathcal{S}_{\delta,L}=\Lambda_L\cap (z_j-\delta,z_j+\delta)$. In the following theorem $H_L$ will denote the restriction of the deterministic operator $-\Delta+V$ where $V:\R\to\R$ is a measurable function.

\begin{theo}\label{SFUC}
Let $\delta\in(0,1/2), K_V\geq0$ and $E\in\R$. Then, there is a constant $C_{sfuc}=C_{sfuc}(\delta,K_V,E)\in(0,\infty)$ such that for all measurable potentials $V:\R\to[-K_V,K_V]$, all scales $L\in\N$ with $L\geq 18e$, all sequences $(z_j)_{j\in\Z}\subset \R^d$ such that $\forall j\in\Z$, $|z_j-j|\leq 1$ and all linear combinations of eigenfuctions 
\begin{displaymath}
\Psi=\sum_{n\in\N\,:\, E_n\leq E}\alpha_n\Psi_n
\end{displaymath} 
(where $\Psi_n$ satisfies $H_L\Psi_n=E_n\Psi_n$ and $\alpha_n\in\C$)
we have
\begin{displaymath}
\int_{S_{L,\delta}}|\Psi|^2\geq C_{sfuc} \int_{\Lambda_L} |\Psi|^2
\end{displaymath}
\end{theo}
We can now apply this theorem to our random operator $H_\omega(\Lambda_L)$ and prove the
\begin{lem}
Fix $E\in\R$. There exists $C>1$ such that for all $L>0$ and any random eigenvalue $E_j(\omega)<E$ of $H_\omega(\Lambda_L)$, $\|\nabla_\omega\big(E_j(\omega))\|_1\in[1/C,C]$.
\end{lem}
\begin{proof}
First, by assumption, there exists an interval of $[z-\delta,z+\delta]$ with $\delta\in(0,1/2)$ included in $(0,1)$ on which $q$ is bounded from below by a constant $\eta>0$. Furthermore, $q$ is also bounded from above by $\dfrac{1}{\eta}$ and supported in $[-N,N]$. Therefore, we have
\begin{displaymath}
\dfrac{1}{\eta}\int_{-N}^N \phi_j^2(t+n)\geq \partial_{\omega_n} E_j=\int_{-N}^N q(t) \phi_j^2(t+n)\geq \eta \int_{(z+n-\delta,z+n+\delta)} \phi_j^2(t)
\end{displaymath}
Therefore, if we set $z_n=z+n$ and define $\mathcal{S}_{\delta,L}$ as above, Theorem~\ref{SFUC} yields
\begin{displaymath}
\dfrac{2N}{\eta}\geq\|\nabla_\omega E_j\|_1\geq \int_{\mathcal{S}_{\delta,L}} \phi_j^2(t)\geq \eta\cdot C_{sfuc}
\end{displaymath}
since $\phi_j$ is a normalized eigenvector associated to $H_\omega(\Lambda_L)$ whose random potential is uniformly bounded in $\omega$.
\end{proof}

 In the rest of the subsection, $M$ will be fixed such that $\mathbb{P}(|\omega_0|>M)=0$ so that all the random variables $(\omega_i)_i$ are almost surely bounded by $M$.

For $\bullet\in\{F,G\}$, define the following ODE
\begin{equation}
(\mathcal{E}^n_\bullet)\,: \,
\forall x\in(-N,N),y''(x)+V_\omega(n+x) y(x) = \bullet y(x)
\end{equation}
Now, fix an q-orthonormal basis $(e_{1,\bullet}^n,e_{2,\bullet}^n)$ of the space of solutions of $(\mathcal{E}^n_\bullet)$. 

\begin{prop}\label{anabasis}
Let $\bullet\in\{F,G\}$. We can choose $e_{1,\bullet}^n$ and $e_{2,\bullet}^n$ so that they are analytic functions of the $(\omega_j)_{j\in\llbracket n-2N,n+2N\rrbracket}\in[-M,M]^{4N+1}$.
\end{prop}
\begin{proof}
We omit the dependence on $n$ and $\bullet$ and only write $\omega$ instead of \\$(\omega_j)_{j\in\llbracket n-N,n+N\rrbracket}$. Let $\Psi$ and $\Phi$ be the solutions of $(\mathcal{E}^n_j)$ satisfying $\Psi'(0)=\Phi(0)=0$ and $\Psi(0)=\Phi'(0)=1$. We know that $\Psi$ and $\Phi$ are power series of $\omega$ and that $\|\Psi\|_q\|\Phi\|_q \neq 0$. Thus, $e_1:=\dfrac{\Psi}{\|\Psi\|_q}$ is analytic and satisfies $(\mathcal{E}^n_j)$. Now, define $\tilde{\Phi}:=\Phi-\langle \Phi,e_1\rangle e_1$. Then, $\tilde{\Phi}$ is an analytic non-zero function orthogonal to $e_1$ satisfying   $(\mathcal{E}^n_j)$. This concludes the proof of Proposition~\ref{anabasis}, taking $e_1$ and $e_2:=\dfrac{\tilde{\Phi}}{\|\tilde{\Phi}\|_q}$.
\end{proof}
 Now, as $u$ satisfies the ODE
\begin{equation}
\forall x\in(-N,N),\,y''+V_\omega(n+x)y(x)=Fy(x)+(E_j(\omega)-F)y(x)
\end{equation}
 with $|E_j(\omega)-F|\leq e^{-l^\beta}$ ($v$ satisfies a similar ODE) there exist two unique couples $(A_n,B_n)\in\R^2$ and $(\tilde{A}_n,\tilde{B}_n)\in\R^2$ such that, for all $x\in(n-N,n+N)$, 
\begin{equation}
\left\{\begin{aligned}
u(x):=A_n e_{1,F}^n(x-n)+B_n e_{2,F}^n(x-n)+\epsilon_u^n(x-n)\\
v(x)=\tilde{A}_n e_{1,G}^n(x-n)+\tilde{B}_n e_{2,G}^n(x-n)+\epsilon_v^n(x-n)
\end{aligned}\right. .
\end{equation} 
 and such that for $\bullet\in\{u,v\}$ we have $\epsilon_\bullet^n(0)=\left(\epsilon_\bullet^n\right)'(0)=0$. We then have 
 \begin{displaymath}
 \|\epsilon_u^n\|_\infty+\|\epsilon_v^n\|_\infty+\|\left(\epsilon_u^n\right)'\|_\infty+\|\left(\epsilon_v^n\right)'\|_\infty\leq Ce^{-l^\beta}
 \end{displaymath}
 for some $C>0$ (depending only on $\|q\|_\infty$, $M$ and $N$). Therefore, \\
 $\|u_{|_{(n-N,n+N)}}\|_q^2=A_n^2+B_n^2+\varepsilon_n^u$ and $\|v_{|_{(n-N,n+N)}}\|_q^2=\tilde{A}_n^2+\tilde{B}_n^2+\varepsilon_n^v$ with $|\varepsilon_n^u|+|\varepsilon_n^v|\leq Ce^{-l^\beta}$. Thus, 
\begin{equation}
\left\{\begin{aligned}
\mathcal{N}:=\|\nabla E_j\|_1=\sum_{n=-l}^l (A_n^2+B_n^2)+\xi_u\\
\tilde{\mathcal{N}}:=\|\nabla E_k\|_1=\sum_{n=-l}^l (\tilde{A}_n^2+\tilde{B}_n^2)+\xi_v
\end{aligned}\right.
\end{equation}
with $|\xi_u|+|\xi_v|\leq Ce^{-l^\beta}$.
Now, define : $\left\{
\begin{aligned}
C_n:=\dfrac{A_n}{\sqrt{\mathcal{N}}}~~,~~\tilde{C}_n:=\dfrac{\tilde{A}_n}{\sqrt{\tilde{\mathcal{N}}}}\\
D_n:=\dfrac{B_n}{\sqrt{\mathcal{N}}}~~,~~\tilde{D}_n:=\dfrac{\tilde{B}_n}{\sqrt{\tilde{\mathcal{N}}}}
\end{aligned}
\right. $. Then, we have
\begin{equation}
\sum_{n=-l}^l C_n^2+D_n^2=\sum_{n=-l}^l \tilde{C}_n^2+\tilde{D}_n^2+O(e^{-l^\beta})=1+O(e^{-l^\beta}).
\end{equation}

Finally, define $U(n)=\begin{pmatrix}
C_n\\
D_n
\end{pmatrix}$ and $V(n)=\begin{pmatrix}
\tilde{C}_n\\
\tilde{D}_n
\end{pmatrix}$, define the Prüfer variables $(r_u,\theta_u)\in\R_+^*\times[0,2\pi)$ such that
$U(n)=r_u\begin{pmatrix}
\sin \theta_u\\
\cos \theta_u
\end{pmatrix}$ 	and define \\
$t_u:=\text{sgn}(\tan \theta_u) \inf\left(|\tan \theta_u|,|\cot \theta_u|\right)$ and the same for $t_v$. The function $t_u$ is equal to $\tan \theta_u$ or $\cot \theta_u$ depending on whether $|\tan \theta_u|\leq 1$ or $|\tan \theta_u|\geq 1$. Using these notations, \eqref{gradcoli} can be rewritten
\begin{equation}
\|r_u-r_v\|_1\leq Ce^{-l^\beta/2}.
\end{equation}

The proof of Lemma~\ref{probcoli} is the exact same as in \cite{S14b}, except for the proof of Lemma~\ref{closetan} below (Lemma 3.8 in \cite{S14b}), and will not be rewritten here, as it is quite technical.
Therefore, we will only prove the 
\begin{lem}\label{closetan}
There exist nine analytic functions $(f_i)_{i\in\llbracket 0,8\rrbracket}$ (only depending on $q$ and $N$) defined on $\R^{4N+1}$ and not all constantly equal to zero such that, if $u$ (respectively $v$) is a 1-normalized eigenfunction of $H_\omega(\Lambda_L)$ associated to $E_j(\omega)\in\left[F-e^{-l^\beta},F+e^{-l^\beta}\right]$ (respectively associated to $E_k(\omega)\in[G-e^{-l^\beta},G+e^{-l^\beta}]$), if for some $n_0\in\Z\cap(-l+2N,l-2N)$ we have $r_u(n_0)\geq e^{-l^\beta/4}$ and
\begin{displaymath}
\forall m\in\llbracket n_0-7N,n_0+7N\rrbracket, \left|r_u(m)-r_v(m)\right|\leq e^{-l^\beta/2} 
\end{displaymath}
and if we define the polynomials 
\begin{displaymath}
\mathcal{R}_{\hat{\omega}}(X):=\sum_{i=0}^8 f_i(\hat{\omega}) X^i\text{ and }
\mathcal{Q}_{\hat{\omega}}(X):=\sum_{i=0}^8 f_{8-i}(\hat{\omega}) X^i
\end{displaymath}
where we have defined $\hat{\omega}:=\left(\omega_{n_0-8N},\dots,\omega_{n_0+8N}\right)$, then we have : 
\begin{align*}
\text{if }\exists\, g\in\{\tan,\cot\}, 
\left\{\begin{aligned}
t_v(n_0)=g(\theta_v(n_0))\\
t_u(n_0)=g(\theta_u(n_0))
\end{aligned}\right., 
\text{ then }&\left|\mathcal{R}_{\hat{\omega}}\left(t_v\left(n_0\right)\right)\right|\leq e^{-l^\beta/4},\\
\text{otherwise, we have }  &\left|\mathcal{Q}_{\hat{\omega}}\left(t_v\left(n_0\right)\right)\right|\leq e^{-l^\beta/4}.
\end{align*}
\end{lem}
\begin{proof}
We will prove the result under the assumption $t_u(n_0)=\tan\theta_u(n_0)$ and $t_v(n_0)=\tan \theta_v(n_0)$, i.e when
\begin{equation}\label{assumptan}
\max\left(|\tan \theta_u(n_0)|,|\tan \theta_v(n_0)|\right)\leq 1.
\end{equation}  
There are minor modifications in the other cases. As the random variables are i.i.d, it suffices to show the result with $n_0=0$, which will be supposed from now on. We then consider the ODE
\begin{equation}
\forall x\in(-7N,7N),y''(x)+V_\omega(x) y(x) = F y(x),
\end{equation}
which depends only on $(\omega_{-8N},\dots,\omega_{8N})$. Suppose $|r_u(m)-r_v(m)|\leq e^{-l^\beta/2}$ for $m\in\llbracket -7N,7N\rrbracket$ and $r_u(0)\geq e^{-l^\beta/4}$. We show that $t_v(0)$ is almost a root of a polynomial depending only on $(\omega_{-8N},\dots,\omega_{8N})$. 

In the following lines $\varepsilon$ will denote a vector such that $\|\varepsilon\|\leq Ce^{-l^\beta/2}$, its value may change from a line to another. As $u$ and $v$ have continuous derivatives, 
\begin{equation}
M^F U(2N)=N^F U(0)+\varepsilon
\end{equation}
where $M^F:=\begin{pmatrix}
e_{1,F}^0(-N) & e_{2,F}^0(-N)\\
(e_{1,F}^0)'(-N) & (e_{2,F}^0)'(-N)
\end{pmatrix}$ and 
$N^F:=\begin{pmatrix}
e_{1,F}^0(N) & e_{2,F}^0(N)\\
(e_{1,F}^0)'(N) & (e_{2,F}^0)'(N)
\end{pmatrix}$. Thus, if we define $T_F^+:= (M^F)^{-1}N^F$ we have
\begin{equation}\label{defT}
U(2N)=T_F^+\,U(0) +\varepsilon\text{  and  } V(2N)=T_F^+\,V(0)+\varepsilon.
\end{equation}
Indeed, the matrix $(M_F)^{-1}$ depends only on $(\omega_{-N},\dots,\omega_N)\in[-M,M]^{2N+1}$ and is therefore uniformly bounded by a constant $C>0$ (depending only on $\|q\|_\infty$, $M$ and $N$).

As $t_u(0)=\tan\theta_u(0)$, we compute  
\begin{align*}
\left(\dfrac{r_u(2N)}{r_u(0)}\right)^2=&\left\|T_F^+\begin{pmatrix}
\sin \theta_u(n)\\
\cos \theta_u(n)
\end{pmatrix}\right\|^2+\epsilon \\
=&\dfrac{1}{1+t_u(n)^2}\left\|T_F^+\begin{pmatrix}
t_u(n) \\
1
\end{pmatrix}\right\|^2+\epsilon,
\end{align*}
for some $|\epsilon|\leq Ce^{-l^\beta/4}$.
In the case $t_u(0)=\dfrac{1}{\tan\theta_u(0)}$, we compute
\begin{align*}
\left(\dfrac{r_u(2N)}{r_u(0)}\right)^2=&\left\|T_F^+\begin{pmatrix}
\sin \theta_u(n)\\
\cos \theta_u(n)
\end{pmatrix}\right\|^2+\epsilon \\
=&\dfrac{1}{1+t_u(n)^2}\left\|T_F^+\begin{pmatrix}
1 \\
t_u(n)
\end{pmatrix}\right\|^2+\epsilon.
\end{align*}

The eigenvector $v$ satisfies the same equation if we replace $F$ by $G$. Therefore, the equation
\begin{displaymath}
\left|\left(\dfrac{r_u(2N)}{r_u(0)}\right)^2-\left(\dfrac{r_v(2N)}{r_v(0)}\right)^2 \right| \leq e^{-l^\beta/4}
\end{displaymath}
 can be rewritten
 \begin{equation}
\left|\dfrac{1}{1+t_u(0)^2}\left\|T_F^+\begin{pmatrix}
t_u(n) \\
1
\end{pmatrix}\right\|^2  \\
- \dfrac{1}{1+t_v(0)^2}\left\|T_G^+\begin{pmatrix}
t_v(n) \\
1
\end{pmatrix}\right\|^2\right|\leq Ce^{-l^\beta/4}.
\end{equation} 
Thus, there exists $\epsilon_1$ such that $|\epsilon_1|\leq Ce^{\-l^\beta/4}$ and such that
\begin{equation}\label{eqabove}
\dfrac{1}{1+t_u(0)^2}\left\|T_F^+\begin{pmatrix}
t_u(n) \\
1
\end{pmatrix}\right\|^2 \\
=\dfrac{1}{1+t_v(0)^2}\left\|T_G^+\begin{pmatrix}
t_v(n) \\
1
\end{pmatrix}\right\|^2+\epsilon_1.
\end{equation}
Now, consider the equation $U(-2N)=T_F^-U(n)+\varepsilon$ for the matrix $T_F^-$ constructed in the same way as $T_F^+$. Using the same calculations as to prove \eqref{eqabove} we obtain the existence of $\eta_1$ with $|\eta_1|\leq Ce^{\-l^\beta/4}$ such that 
\begin{equation}\label{eqbelow}
\dfrac{1}{1+t_u(0)^2}\left\|T_F^-\begin{pmatrix}
t_u(0) \\
1
\end{pmatrix}\right\|^2 \\
=\dfrac{1}{1+t_v(0)^2}\left\|T_G^-\begin{pmatrix}
t_v(0) \\
1
\end{pmatrix}\right\|^2+\eta_1.
\end{equation}
Define the polynomials of degree $2$
\begin{displaymath}
\begin{aligned}
P_G(t)&:=\left\|T_G^+\begin{pmatrix}
t \\
1
\end{pmatrix}\right\|^2
\end{aligned}
\text{ and  }
\begin{aligned}
Q_G^n&:=\left\|T_G^-\begin{pmatrix}
t \\
1
\end{pmatrix}\right\|^2.
\end{aligned} 
\end{displaymath}
Using \eqref{assumptan}, the equations \eqref{eqabove} and \eqref{eqbelow} can be rewritten
\begin{multline}\label{eqplus}
R_1(t_u(0),t_v(0)):=t_u^2(0)\left[(1+t_v^2(0))\left\|T_F^+\begin{pmatrix}
1 \\
0
\end{pmatrix}\right\|^2-P_G(t_v(0))\right] \\
+2t_u(0)\left\langle T_F^+\begin{pmatrix}
1 \\
0
\end{pmatrix},T_F^+\begin{pmatrix}
0 \\
1
\end{pmatrix} \right\rangle(1+t_v^2(0)) \\
+\left[(1+t_v^2(0))\left\|T_F^+\begin{pmatrix}
0 \\
1
\end{pmatrix}\right\|^2-P_G(t_v(0))\right]=\epsilon_2
\end{multline}
and 
\begin{multline}\label{eqminus}
R_2(t_u(0),t_v(0)):=t_u^2(n)\left[(1+t_v^2(0))\left\|T_F^-\begin{pmatrix}
1 \\
0
\end{pmatrix}\right\|^2-P_G(t_v(0))\right] \\
+2t_u(0)\left\langle T_F^-\begin{pmatrix}
1 \\
0
\end{pmatrix},T_F^-\begin{pmatrix}
0 \\
1
\end{pmatrix} \right\rangle(1+t_v^2(0)) \\
+\left[(1+t_v^2(0))\left\|T_F^-\begin{pmatrix}
0 \\
1
\end{pmatrix}\right\|^2-P_G(t_v(0))\right]=\eta_2.
\end{multline}

Thus, $t_u(0)$ is a root of the two polynomials $t\to R_1(t,t_v(0))-\epsilon_2$ and $t\to R_2-\eta_2$. Therefore, the resultant of these polynomials must be zero. 
All the coefficients in $R_1$ and $R_2$ are bounded uniformly over $(\omega_m)_{m\in\llbracket -N,N\rrbracket}$. Thus the resultant $\mathcal{R}(t_v(0))$ of $R_1(\,\cdot\,,t_v(0))$ and $R_2(\,\cdot\,,t_v(0))$ is smaller than $e^{-l^\beta/4}$. 

If we have $t_u(0)=\tan\theta_u(0)$ but $t_v(0)=\cot\theta_v(0)$ instead of $t_v(0)=\tan\theta_v(0)$, the resultant $ \mathcal{Q}(t)$ obtained is equal to $t^8\mathcal{R}(1/t)$.  If we have $t_u(0)=\cot\theta_u(0)$ and $t_u(0)=\tan\theta_u(0)$ instead of $t_u(0)=\tan\theta_u(0)$ and $t_v(0)=\tan\theta_v(0)$, the resultant obtained is $\mathcal{R}$

Now, the resultant $\mathcal{R}$ is an analytic function of the random variables \\$\left(\omega_{-2N},\omega_{-2N+1},\dots,\omega_{2N}\right)$. We will now prove that, as a function of these random variables, it is not constantly the zero polynomial. This will be done under the assumption $\omega_{-2N}=\omega_{-8N+1}=\dots=\omega_{2N}$. Under this assumption we have 
\begin{displaymath}
\forall x\in(-N,N),~ V_\omega(x)=\sum_{n\in\Z} \omega_n q(x-n)=\omega_0 \sum_{n\in(x-N,X+N)\cap\Z}q(x-n).
\end{displaymath}
 Therefore, we come down to the study of the ODEs
\begin{equation}\label{ODE}
\forall x\in(-N,N),y''(x)=\left(\omega_0 \tilde{q}(x)-\bullet\right) y(x)
\end{equation}
where $\tilde{q}$ is one-periodic, $\tilde{q}>0$ on some interval $\mathcal{K}\subset(-1/2,1/2)$ and $\bullet\in\{F,G\}$. 

We now prove the
\begin{lem}\label{resnotzero}
There exists $\omega_{0}$ such that, $\mathcal{R}_{i_0}$ is not the zero polynomial.
\end{lem}

\begin{proof}
The fact that the resultant $\mathcal{R}_{i_0}$ is the zero polynomial is equivalent to the fact that, for all $t'\in\R$, the polynomials $R_1(\,\cdot\, ,t')$ and $R_2(\,\cdot\, ,t')$ have a common root. This is also equivalent to the fact that for all $w\neq0$ satisfying $\eqref{ODE}$ for $\bullet=G$, there exists a function $z:=z(w)\neq 0$ satisfying \eqref{ODE} for $\bullet=F$ such that 
\begin{equation}\label{cond}
\left|\left(\dfrac{r_z(2N)}{r_z(0)}\right)^2-\left(\dfrac{r_w(2N)}{r_w(0)}\right)^2 \right|+ \left|\left(\dfrac{r_z(-2N)}{r_z(0)}\right)^2-\left(\dfrac{r_w(-2N)}{r_w(0)}\right)^2 \right|=0.
\end{equation}

Now, using Propositions~\ref{instability}, \ref{discnotequal} and \ref{discnotequalpos}, we know there exists $\omega_0$ such that the discriminants $D_F$ and $D_G$ are not equal or opposite and satisfy $|D_F|>2$ and $|D_G|>2$. Let $\lambda_F$ and $\lambda_G$ be the associated Floquet multipliers with absolute value strictly larger than $1$. Without loss of generality we will suppose that $|\lambda_F|>|\lambda_G|>1$. Let $w$ be a normalized Floquet solution associated to $\lambda_F$. Then, $\dfrac{r_w(2N)}{r_w(0)}=\dfrac{r_w(0)}{r_w(-2N)}=\lambda_F^2$. Let $(f_1^G,f_2^G)$ be normalized Floquet solutions associated to $\lambda_G$ and $\lambda_G^{-1}$ and let $p_G=\langle f_1^G,f_2^G\rangle_q$. Then, for any solution $z$ of \eqref{ODE} with $\bullet=G$, there exists $(A,B)\in\R^2$ such that
\begin{displaymath}
z=Af_1^G+Bf_2^G.
\end{displaymath}
Therefore, we compute
\begin{align*}
T_G^+z=A\lambda_Gf_1^G+B \lambda_G^{-1} f_2^G\\
T_G^-z=A\lambda_G^{-1}f_1^G+B \lambda_G f_2^G
\end{align*}
and
\begin{align*}
\dfrac{r_z(2N)}{r_z(0)}=\dfrac{A^2\lambda_G^2+B^2 \lambda_G^{-2}+2p_GAB}{A^2+B^2+2p_GAB}\\
\dfrac{r_z(-2N)}{r_z(0)}=\dfrac{A^2\lambda_G^{-2}+B^2 \lambda_G^2+2p_GAB}{A^2+B^2+2p_GAB}
\end{align*}
Let $(A,B)=R(\sin \phi,\cos\phi)$ with $(R,\phi)\in(0,\infty)\times[0,2\pi)$. Then, for $w$ as above \eqref{cond} can only be satisfied for $z$ such that $B\neq 0$. Let $\tau=\tan \phi$. Then, we have
\begin{align*}
\dfrac{r_z(2N)}{r_z(0)}=\dfrac{\tau^2\lambda_G^2+\lambda_G^{-2}+2p_G\tau}{\tau^2+1+2p_G\tau}\\
\dfrac{r_z(-2N)}{r_z(0)}=\dfrac{\tau^2\lambda_G^{-2}+\lambda_G^2+2p_G\tau}{\tau^2+1+2p_G\tau}
\end{align*}
Therefore, there exists a non-zero $z$ such that \eqref{cond} is satisfied if and only if there exits $\tau\in\R$ such that 
\begin{align*}
\lambda_F^2=\dfrac{\tau^2\lambda_G^2+\lambda_G^{-2}+2p_G\tau}{\tau^2+1+2p_G\tau}\\
\lambda_F^{-2}=\dfrac{\tau^2\lambda_G^{-2}+\lambda_G^2+2p_G\tau}{\tau^2+1+2p_G\tau}
\end{align*}
This is also equivalent to the fact that there exists $\tau\in\R$ such that
\begin{align*}
\tau^2\left(\lambda_G^2-\lambda_F^2\right)+2p_G(1-\lambda_F^2)\tau+\left(\lambda_G^{-2}-\lambda_F^2\right)=0\\
\tau^2\left(\lambda_G^{-2}-\lambda_F^{-2}\right)+2p_G(1-\lambda_F^{-2})\tau+\left(\lambda_G^2-\lambda_F^{-2}\right)=0
\end{align*}
Eventually, this is also equivalent to the fact that the resultant of these two polynomials is zero, i.e

\begin{equation}
\begin{vmatrix}
\lambda_G^2-\lambda_F^2 & 0 & \lambda_G^{-2}-\lambda_F^{-2} & 0\\
2p_G(1-\lambda_F^2) & \lambda_G^2-\lambda_F^2 & 2p_G(1-\lambda_F^{-2}) & \lambda_G^{-2}-\lambda_F^{-2}\\
\lambda_G^{-2}-\lambda_F^2 & 2p_G(1-\lambda_F^2) &  \lambda_G^2-\lambda_F^{-2} & 2p_G(1-\lambda_F^{-2}) \\
0 & \lambda_G^{-2}-\lambda_F^2 & 0 & \lambda_G^2-\lambda_F^{-2}
\end{vmatrix}=0
\end{equation}
Now, define
\begin{displaymath}
\left\{
\begin{aligned}
\Delta_1&=\lambda_G^2-\lambda_F^2<0\\
\Pi_+&=2p_G(1-\lambda_F^2)\\
\Delta_2&=\lambda_G^{-2}-\lambda_F^2<0\\
\Delta_3&=\lambda_G^{-2}-\lambda_F^{-2}=-\Delta_1\lambda_G^{-2}\lambda_F^{-2}\\
\Pi_-&=2p_G(1-\lambda_F^{-2})=-\Pi_+\lambda_F^{-2}\\
\Delta_4&=\lambda_G^2-\lambda_F^{-2}=-\Delta_2\lambda_G^2\lambda_F^{-2}
\end{aligned}
\right.
\end{displaymath}
To conclude the proof of Lemma~\ref{resnotzero} it therefore suffices to show that the matrix
\begin{equation}
M:=\begin{pmatrix}
\Delta_1 & 0 & \Delta_3 & 0\\
\Pi_+ & \Delta_1 & \Pi_- & \Delta_3\\
\Delta_2 & \Pi_+ &  \Delta_4 & \Pi_- \\
0 & \Delta_2 & 0 & \Delta_4
\end{pmatrix}
\end{equation}
satisfies $\det(M)\neq 0$.

A straightforward calculus shows that
\begin{align*}
\det(M)&=(\Delta_1\Delta_4-\Delta_2\Delta_3)^2+\left(\Pi_-\Delta_1-\Pi_+\Delta_3\right)\left(\Pi_-\Delta_2-\Pi_+\Delta_4\right)\\
&=[\Delta_1\Delta_2\lambda_F^{-2}(\lambda_G^2-\lambda_G^{-2})]^2+\Pi_+^2\Delta_1\Delta_2\left(-\lambda_F^2+\lambda_F^{-2}\lambda_G^{-2}\right)\left(-\lambda_F^2+\lambda_G^2\lambda_F^{-2}\right) =:A^2+B
\end{align*}

with
\begin{displaymath}
A=\Delta_1\Delta_2\lambda_F^{-2}(\lambda_G^2-\lambda_G^{-2})\neq 0
\end{displaymath}

Now, we compute

\begin{displaymath}
\left(-\lambda_F^2+\lambda_F^{-2}\lambda_G^{-2}\right)\left(-\lambda_F^2+\lambda_G^2\lambda_F^{-2}\right)=\lambda_F^{-4}\left(\lambda_F^{4}-\lambda_G^{-2}\right)\left(\lambda_F^{4}-\lambda_G^2\right)>0
\end{displaymath}
because $\lambda_F^2>\lambda_G^2>1$. Now, as $\delta_1\delta_2>0$, we have $A^2>0$ and $B\geq 0$. Therefore, $\det(M)>0$ and the resultant of the two polynomial is not zero for our choice of $\omega_0$. This concludes the proof of Lemma~\ref{resnotzero}.
 \end{proof}
 
 We can now finish the proof of Lemma~\ref{closetan}. There exist $(\omega_{-8N},\dots,\omega_{8N})$ such that the coefficients of $\mathcal{R}_{i_0}$ are not all equal to zero. Now, write \\
 $\mathcal{R}_{i_0}(X)=\sum_{i=0}^8 f_i(\omega_{-8N},\dots,\omega_{8N})X^i$ where the $(f_i)_i$ are analytic. Then, one of the functions $(f_i)_i$  must be not constantly equal to zero. Besides, by construction, we have $|\mathcal{R}_{i_0}(t_v(0))|\leq e^{-l^\beta/4}$. This completes the proof of Lemma~\ref{closetan} and therefore the proof of Lemma~\ref{probcoli}, as in \cite{S14b}.
 \end{proof}

\appendix
\section{Analytic functions of several real variables}
In this section, we extract two properties of analytic functions of several real variables that were proved in \cite{S14b}, but first, we remind the reader of the Weierstrass preparation theorem (\cite{L65}). For $x\in\R^n$, we will write $x=(\hat{x},x_n)$. 
\begin{theo}\label{Weierstrass}
Let $\mathcal{O}$ be an open subset of $\R^n$ that contains the origin and let $f:\mathcal{O}\to\R$ be an analytic function vanishing at the origin such that the analytic function $x_n\mapsto f(0,\cdots,0,x_n)$ has a zero of order $m\in\N^*$ at $0$. There exists a neighborhood $\mathcal{U}$ of the origin, a Weierstrass polynomial $P(\hat{x},x_n)=a_0(\hat{x})+ a_1(\hat{x})x_n+\cdots+a_{m-1}(\hat{x})x_n^{m-1}+x_n^m$, defined on $\mathcal{U}$, with $a_i(0)=0$ for all $i\in\llbracket 1,m-1\rrbracket$, and an analytic function $g:\mathcal{U}\to\R$ with $g(0)\neq 0$, such that, for all $x\in\mathcal{U}$, we have $f(x)=P(x)g(x)$.
\end{theo}

\begin{prop}\label{Loja}
Fix $\Omega\subset \R^n$ an open set and $f:\Omega\to \R$ a non zero analytic function. Fix $G$ a compact subset of $\Omega$. There exist $\epsilon_0>0$ and $m\in\N^*$ such that, for all $0<\epsilon<\epsilon_0$, we have $\left|\{x\in G, |f(x)|<\epsilon\}\right |\leq  \left(\dfrac{\epsilon}{\epsilon_0}\right)^{1/m}$.
\end{prop}

\begin{prop}\label{Loja2}
Fix $\Omega\subset \R^n$ an open set containing the origin and $f:\Omega\to \R$ an analytic function such that, for all $(\hat{x},x_n)\in\Omega$, the function $h_n\mapsto f(\hat{x},x_n+h_n)$ is not constantly equal to zero in a neighborhood of the $0$. Fix $G:=[-M,M]^n$ a compact subset of $\Omega$. There exists $\epsilon_0>0$ and $m\in\N^*$ such that, for all $0<\epsilon<\epsilon_0$ and $\hat{x}\in[-M,M]^{n-1}$, we have $\left|\{x_n\in[-M,M], |f(\hat{x},x_n)|<\epsilon\}\right |\leq  \left(\dfrac{\epsilon}{\epsilon_0}\right)^{1/m}$.
\end{prop}

\section{Properties of Floquet solutions}
In this section, we remind the reader of fact concerning Floquet theory and prove simple properties that are used in the proof of decorrelation estimates.

Let $W:\R\to\R$ be a 1-periodic bounded potential and $w:\R\to\R$ a 1-periodic non-negative weight function. For $\lambda\in\R$, consider the following ODE
\begin{equation}\label{ODEflo}
y''(x)+W(x)y(x)=\lambda w(x)y(x)
\end{equation}
Let $\Phi_\lambda,\Psi_\lambda$ be the solutions of this ODE satisfying $\Phi_\lambda(0)=\Psi'_\lambda(0)=1$ and $\Phi_\lambda'(0)=\Psi_\lambda(0)=0$.
Now, define the matrix $T(\lambda)=\begin{pmatrix}
\Phi_\lambda(1)& \Psi_\lambda(1)\\
\Phi'_\lambda(1)& \Psi'_\lambda(1)
\end{pmatrix}$, and $D(\lambda):=tr[T(\Lambda)]$. We know that $\det(T(\lambda))=1$, so the characteristic polynomial of $T(\lambda)$ is $X^2-D(\lambda)X+1$.

If $D(\lambda)=\pm2$, $1$ is the only eigenvalue of $T$, and there exists two different solutions $(u,v)$ of \eqref{ODEflo} that satisfy
\begin{align*}
u(x+1)&=\pm u(x)\\
v(x+1)&=(ax\pm 1)v(x)
\end{align*}
for some $a\in\R$. Note that $a=0$ if and only if $T(\lambda)=I_2$, and that if $T(\lambda)\neq I_2$, two such solutions differs from a multiplicative scalar. We then obtain at least one periodic solution when $D(\lambda)=2$ and one semi-periodic solution when $D(\lambda)=-2$

If $D(\lambda)\neq 2$, and if we note $\mu_{\pm}(\lambda)=\dfrac{D(\lambda)\pm\sqrt{D(\lambda)^2-4}}{2}$, which are named Floquet multipliers, then there exists two solutions $(u,v)$ that satisfy
\begin{align*}
u(x+1)=\mu_+(\lambda)u(x)\\
v(x+1)=\mu_-(\lambda)v(x)
\end{align*}
and any other solution satisfying one of these conditions differ of a multiplicative scalar. These solutions are called Floquet solutions.  

Now, fix $(E_1,E_2)\in\R^2$ with $E_1<E_2$, $q_{per}$ a bounded 1 periodic-function and $w$ a 1-periodic bounded weight function such that there exists $\eta>0$ and $\mathcal{K}\subset [0,1]$ an interval satisfying
\begin{equation}
\eta \cdot 1_{\mathcal{K}}\leq w
\end{equation}

For $i\in\{1,2\}$ and $\lambda\in\R$ we consider the ODE
\begin{displaymath}
 (\mathcal{E}_i^\lambda):\,y''+(q_{per}+E_i) y= \lambda w y
\end{displaymath}
and we consider $D_1(\lambda)$ and $D_2(\lambda)$ as defined above. We will suppose either that $w$ is bounded from below by a positive constant or that $q_{per}:=0$.
We now prove the
\begin{prop}\label{instability}
There exists $\lambda_0$ such that for $\lambda<\lambda_0$ we have $|D_i(\lambda)|>2$ for $i\in\{1,2\}$
\end{prop}
\begin{proof}
It suffices to prove the result for $D_1$, for instance. Because $w$ is non-negative and positive on a interval a positive length, the ODE are Sturm-Liouville equations in the so-called semi-definite case. Therefore, following \cite[Section 4]{BW13}, there exists $\lambda_{inf}$ such that $(-\infty,\lambda_{inf})$ is an instability interval, so that every $\lambda<\lambda_{inf}$ satisfies $|D_1(\lambda)|>2$.
\end{proof}

\begin{prop}\label{discnotequalpos}
Suppose $w$ is bounded from below by a positive constant. Then, there exists $\lambda$ such that $|D_1(\lambda)|\neq |D_2(\lambda)|$
\end{prop}
\begin{proof}
Since, $w$ is bounded from below by a positive constant, we are in the so-called definite case. We can consider the eigenvalue problem
\begin{equation}
H_iy:=\dfrac{1}{w}(y''+(q_{per}+E_i)y)=\lambda y
\end{equation}
As $E_1<E_2$, the lower anti-periodic eigenvalue $\lambda_1$ of $H_1$ is strictly smaller than the lower periodic eigenvalue $\lambda_2$ of $H_2$. Therefore, $D_1(\lambda_2)>2$ whereas $D_2(\lambda_2)=2$.
\end{proof}

\begin{prop}\label{discnotequal}
Suppose $q_{per}:=0$. There exists a set $\mathcal{S}\subset \R$ with no accumulation point such that, if $(E_1,E_2)\in\R^2-\mathcal{S}^2$, there exists $\lambda$ such that $|D_1(\lambda)|\neq |D_2(\lambda)|$
\end{prop}
\begin{proof}
For $i\in\{1,2\}$, let $\Phi_\lambda^i$, $\Psi_\lambda^i$ be the solutions of $H_iy=\lambda y$ satisfying $\Phi^i_\lambda(0)=(\Psi^i_\lambda)'(0)=1$ and $(\Phi^i_\lambda)'(0)=\Psi^i_\lambda(0)=0$. Then, we know from Duhamel formula that
\begin{align*}
\Phi_\lambda^i(t)&=\cos(\sqrt{E_i}t)-\dfrac{\lambda}{\sqrt{E_i}}\int_0^t \sin(\sqrt{E_i}(t-s))w(s)\Phi_\lambda(s)ds,\\
\Psi_\lambda^i(t)&=\dfrac{\sin(\sqrt{E_i}t)}{\sqrt{E_i}}-\dfrac{\lambda}{\sqrt{E_i}}\int_0^t \sin(\sqrt{E_i}(t-s))w(s)\Phi_\lambda(s)ds.
\end{align*}
Therefore, we obtain the first order Taylor expansions
\begin{align*}
\Phi_\lambda^i(1)=\cos(\sqrt{E_i})-\dfrac{\lambda}{\sqrt{E_i}}\int_0^1 \sin(\sqrt{E_i}(1-s))w(s)\cos(\sqrt{E_i}s)ds+o(\lambda),\\
(\Psi_\lambda^i)'(1)=\cos(\sqrt{E_i})-\dfrac{\lambda}{\sqrt{E_i}}\int_0^1 \cos(\sqrt{E_i}(1-s))w(s)\sin(\sqrt{E_i}s)ds+o(\lambda).
\end{align*}
Hence,

\begin{equation}
D_i(\lambda)=2\cos(\sqrt{E_i})-\lambda\left(\int_0^1 w(s)ds\right)\dfrac{\sin\sqrt{E_i}}{\sqrt{E_i}}+o(\lambda)
\end{equation}

Now, as $\int_0^1 w(s)ds\neq 0$ by assumption, we have
\begin{align*}
\left(\forall \lambda,\,D_1(\lambda)=D_2(\lambda)\,\right)& \Longrightarrow \left\{
\begin{aligned}
\cos\sqrt{E_1}=\cos\sqrt{E_2}\\
\dfrac{\sin\sqrt{E_1}}{\sqrt{E_1}}=\dfrac{\sin\sqrt{E_2}}{\sqrt{E_2}}
\end{aligned}
\right.\\
& \Longrightarrow  \left\{
\begin{aligned}
E_1=(k_1\pi)^2\\
E_2=(k_2\pi)^2
\end{aligned}
\right. \text{for some integers } k_1,k_2\text{ with same parity}.
\end{align*}

\end{proof}
\bibliographystyle{plain}
\bibliography{biblio2}
\end{document}